\documentclass[3p,preprint,11pt]{elsarticle}

\usepackage{amssymb,amsmath,amsthm,bm}
\usepackage[normalem]{ulem}
\usepackage{multirow}
\usepackage[retainorgcmds]{IEEEtrantools}

\newtheorem{property}{Property}

\newtheorem{proposition}{Proposition}

\usepackage{graphicx}
\graphicspath{{figures/}}

\newlength\figwidth
\newlength\imagewidth
\usepackage{color}
\definecolor{dgreen}{rgb}{0,.6,0}

\begin{document}

\begin{frontmatter}

\title{Cryptanalyzing a chaos-based image encryption algorithm using alternate structure}

\author[cn-xtu-ma,cn-xtu-icip]{Yu Zhang}
\author[cn-xtu-cie]{Chengqing Li\corref{corr}}
\ead{chengqingg@gmail.com}

\author[hk-cityu]{Kwok-Wo Wong}

\author[cn-xtu-ma]{Shi Shu}

\author[hk-cityu]{Guanrong Chen}

\cortext[corr]{Corresponding author.}

\address[cn-xtu-ma]{School of Mathematics and Computational Science,
Xiangtan University, Xiangtan 411105, Hunan, China}

\address[cn-xtu-icip]{MOE (Ministry of Education) Key Laboratory of Intelligent Computing and Information Processing, Xiangtan University, China}

\address[cn-xtu-cie]{College of Information Engineering,
Xiangtan University, Xiangtan 411105, Hunan, China}

\address[hk-cityu]{Department of Electronic Engineering, City University of Hong Kong, Hong Kong}

\begin{abstract}
Recently, a chaos-based image encryption algorithm using alternate structure (IEAS) was
proposed. This paper focuses on differential cryptanalysis of the algorithm and finds that
some properties of IEAS can support a differential attack to recover equivalent secret key with
a little small number of known plain-images. Detailed approaches of the cryptanalysis for cryptanalyzing IEAS
of the lower round number are presented and the breaking method can be extended to the case of higher round number.
Both theoretical analysis and experiment results are provided to support vulnerability of IEAS against differential attack.
In addition, some other security defects of IEAS, including insensitivity with respect to changes of plain-images and insufficient
size of key space, are also reported.
\end{abstract}

\end{frontmatter}

\section{Introduction}

Security of multimedia data including image and video become more and more important as transmission of multimedia data occurs more and more frequently in the current digital world. However, the big differences between multimedia data and text, such as bulk size of multimedia data and strong redundancy existing in neighboring elements of its uncompressed version, make the traditional text encryption algorithms like DES (Data Encryption Standard) can not protect multimedia data efficiently. In addition, multimedia encryption has other special requirements, like fast encryption speed and easy cascade with the whole system. So, designing specific multimedia encryption algorithm become an urgent task. Meanwhile, chaos theory was developed in depth in the 1960s. The most famous character of chaos is so-called ``butterfly effect", i.e., states of a chaos system are very sensitive to changes of its initial conditions and control parameters. This character is very similar to the confusion and diffusion property of a cryptosystem measuring sensitivity of encryption results with respect to change of the secret key and the plaintext. The subtle similarity inspired researchers design secure multimedia encryption algorithms by combing chaos and cryptography.

Due to simple syntax of uncompress image and easy extension of image encryption scheme to other multimedia data, most chaos-based
multimedia encryption scheme consider image data as encryption object. In the past decade, hundreds of chaos-based image encryption schemes
have been proposed \cite{Fridrich:ChaoticImageEncryption:IJBC98,YaobinMao:CSF2004}. In general, the usage of chaos in designing image encryption
schemes can be categorized as the following three classes:
\begin{itemize}
\item creating position permutation matrices \cite{Fridrich:ChaoticImageEncryption:IJBC98,Wang:BreakMao:PLA05,SolakErcan:AnaFridrich:BAC10,Lcq:Optimal:SP11};

\item generating pseudo-random bit sequence, which is then used to control combination and composition
of some basic arithmetical operations like modulo addition and exclusive or operation \cite{Lisj:BRIE:IEEESCS02,LiShujun:ICIP2002,Li:AttackingRCES2004,Rhouma:BreakLian:PLA08,GaLsj:AnaNCA:SANS09,
LiLi:IVC2009b,Lcq:YenChenWu:SAS10,SolakErcan:AnaFridrich:BAC10}.

\item producing ciphertext directly when plain-bytes of image are converted to initial condition and control parameters
of a chaotic map \cite{DavidLsj:Securityofchaotic:CHAOS08,SolakErcan:Algebraic:IS11}.
\end{itemize}
Some general rules about evaluating security of chaos-based encryption algorithms can be found in \cite{AlvarezLi:Rules:IJBC2006}.

In \cite{Zhang:ImageCrypt:SCSF07}, a new image encryption algorithm using alternate structure (IEAS) based on the general cat-map and OCML (One-way Coupled Map Lattice) was proposed, where the two maps are used for realizing position permutation/diffusion and value substitution respectively. Essentially, structure of IEAS belongs to Feistel networks, i.e., an iterated block cipher where the output of the current round is determined by that of the previous one. This paper focuses on security analysis of IEAS and founds that some properties of IEAS, existing when its integer parameter is even, can be used to support a differential attack to recover equivalent secret key of IEAS with a little number of known/chosen plain-images. The detailed approaches of the differential attack are presented in detail when the round number of IEAS is less than or equal to four. In addition, the cryptanalysis also find some other security defects of IEAS, like insensitivity with respect to changes of plain-images and insufficiently large key space.

The rest of this paper is organized as follows. The next section introduces the image encryption algorithm under study, IEAS, briefly. Section~3 present the comprehensive cryptanalysis on the algorithm with some experiment results. The last section concludes the paper.

\section{IEAS encryption algorithm}
\label{sec:scheme}

The plain-image encrypted by IEAS encryption algorithm is a gray-scale image of size $N\times 2N$ (height$\times$width), which can be denoted by an $N\times 2N$ matrix in domain $\mathbb{Z}_{256}$. The encryption algorithm divides the plain-image into two parts of the same size: $\bm{L}=[L(i,j)]_{i=0, j=0}^{N-1,N-1}$ and $\bm{R}=[R(i,j)]_{i=0, j=0}^{N-1,N-1}$. The corresponding cipher-image is composed of two parts also: $\bm{l}=[l(i,j)]_{i=0, j=0}^{N-1,N-1}$ and $\bm{r}=[r(i,j)]_{i=0, j=0}^{N-1,N-1}$. With these notations, IEAS encryption algorithm can be described as follows\footnote{To make the presentation more concise and complete, some notations in the original paper \cite{Zhang:ImageCrypt:SCSF07} are modified, and some details about the algorithm are also supplied or corrected under precondition that its security is not influenced.}.

\begin{itemize}
\item \textit{The secret key}: the number of iteration round $T$ and the initial condition $K_0\in (0,1)$ of the chaotic Logistic map
\begin{equation*}
\mathrm{f}(x)=\mu\cdot x\cdot(1-x).
\end{equation*}

\item \textit{The initialization procedures}:
\par 1) run the Logistic map iteratively with fixed control parameter, $\mu=4$, $T+2$ times from $K_0$ to generate a chaotic sequence $\{x_{l}\}_{l=0}^{T+1}$. Then, a $32$-bit integer sequence $\{K_{l}\}_{l=0}^{T+1}$ is obtained from $\{x_{l}\}_{l=0}^{T+1}$ as
\begin{equation*}
K_{l}=\lfloor x_{l}\cdot(2^{32}-1)\rfloor.
\label{eq:Interger}
\end{equation*}

\par 2) permute and expand the $32$ binary bits of each element of $\{K_{l}\}_{l=0}^{T+1}$ by the look-up table shown in Table~\ref{tb:expansionpermute} and get a $50$-bit integer sequence $\{K_{l}^*\}_{l=0}^{T+1}$.

\begin{table}[!htb]
\centering \caption{The Expansion Permutation Table.}
\begin{tabular}{c |*{8}{c|}c}
\hline   32 & 1  & 2  & 3  & 4  & 5  & 4  & 5  & 6  & 7 \\
\hline   8  & 9  & 8  & 9  & 10 & 11 & 12 & 13 & 14 & 15\\
\hline   16 & 17 & 16 & 17 & 15 & 16 & 17 & 18 & 19 & 20\\
\hline   21 & 20 & 21 & 22 & 23 & 24 & 25 & 24 & 25 & 26\\
\hline   27 & 28 & 29 & 28 & 29 & 30 & 31 & 32 & 1  & 31\\
\hline
\end{tabular}
\label{tb:expansionpermute}
\end{table}

\par 3) generate $T$ permutation matrixes $\bm{P}_0\sim \bm{P}_{T-1}$, whose every entry represents its sole location in the permuted version of the object to be permuted, as follows. For $l=0\sim T-1$, $i=0\sim N-1$, $j=0\sim N-1$, do
\begin{equation}
\bm{P}_l(i, j)=\bm{C}_l\cdot
\left(
\begin{matrix}
i \\
j
\end{matrix}\right)\bmod N,
\label{eq:permutematrix}
\end{equation}
where $\bm{C}_l$ is the $t$-th element in the matrix set
\begin{equation}
\left\{\left(
\begin{matrix}
1 & a \\
b & ab+1
\end{matrix}\right),
\left(
\begin{matrix}
ab+1 & a \\
b    & 1
\end{matrix}\right),
\left(
\begin{matrix}
a    & 1 \\
ab-1 & b
\end{matrix}\right),
\left(
\begin{matrix}
a & ab-1 \\
1 & b
\end{matrix}\right)\right\},
\label{set:matrix}
\end{equation}
$t=\sum_{k=0}^1 K_{l, k}^*\cdot 2^k$, $a=\sum_{k=0}^7K_{l, k+2}^*\cdot 2^k$, $b=\sum_{k=0}^{7} K_{l, k+10}^*\cdot 2^k$, $K_{l}^*=\sum_{k=0}^{49}K_{l, k}^*\cdot 2^k$.

\par 4) produce $T+2$ mask matrixes, $\bm{V}_0\sim \bm{V}_{T+1}$, of size $N\times N$  with the following two steps.

\begin{itemize}
\item Utilize an OCML model to generate $T+2$ pseudo-random number matrixes of size $N\times N$, $\bm{W}_0\sim \bm{W}_{T+1}$. For $i=0\sim N-1$, $j=0\sim N-1$, do
\begin{equation*}
\bm{W}_l(i, j)=(1-\varepsilon)\cdot f(\bm{W}_l(i, j-1))+\varepsilon\cdot f(\bm{W}_l(i-1, j-1)),
\label{eq:ocml}
\end{equation*}
where $\varepsilon=0.875$, and the boundary conditions, $\bm{W}_l(-1, -1)\sim \bm{W}_l(-1, N-1)$ and $\bm{W}_l(0, -1)\sim \bm{W}_l(N-2, -1)$, are assigned by the chaotic states obtained by iterating the Logistic map $2N$ times from initial condition $(\sum_{k=0}^{31} K_{l, k+18}^*\cdot 2^k)/2^{32}$.

\item Discretize $\bm{W}_0\sim \bm{W}_{T+1}$ into $\bm{V}_0\sim \bm{V}_{T+1}$. For $i=0\sim N-1$, $j=0\sim N-1$, do
\begin{equation*}
\bm{V}_l(i, j)=\lfloor \bm{W}_l(i, j)\cdot 256\rfloor.
\label{eq:discrete}
\end{equation*}
\end{itemize}

\item \textit{The encryption procedure} is composed of $T$ rounds of five main steps. Let $\bm{L}_l$ and $\bm{R}_l$ denote the left half part and the right half part of intermediate data obtained in the $l$-th round of encryption, respectively. The schematic structure of IEAS is shown in Fig.~\ref{fig:originstructure}. Set $l=0$, $\bm{L}_l=\bm{L}$ and $\bm{R}_l=\bm{R}$, IEAS runs with the following five steps repeatedly.
\begin{itemize}
\item \textit{Step (a) mask substitution on the left half part in the current round}: let $l=l+1$, and do
\begin{equation}
\bm{R}_{l}(i,j) = \bm{V}_{l-1}(i,j)\oplus\bm{L}_{l-1}(i,j)
\label{eq:substitutionxor}
\end{equation}
for $i=0\sim N-1$, $j=0\sim N-1$.

\item \textit{Step (b) permutation on the right half part in the next round}: for $i=0\sim N-1$, $j=0\sim N-1$, do
\begin{equation*}
\widetilde{\bm{R}}_{l}(i,j) = \bm{R}_{l}(\bm{P}_{l-1}(i,j)).
\label{eq:substitution}
\end{equation*}
For simplicity, $\bm{R}_{l}(\bm{P}_{l-1})$ denotes this operation in remainder of this paper.

\item \textit{Step (c) substitution on the permuted right part}:
for $k=1\sim N^2-1$, do
\begin{equation}
\bm{L}_{l}( i, j) = \bm{R}_{l-1}(i, j)\oplus \mathrm{g}\left( \widetilde{\bm{R}}_{l}(i, j), \widetilde{\bm{R}}_{l}( i', j' )\right),
\label{eq:substitution}
\end{equation}
where $\bm{L}_{l}( 0, 0)=\bm{R}_{l-1}(0, 0)\oplus \widetilde{\bm{R}}_{l}(0, 0)$, $i=\lfloor k/N \rfloor$, $j=\bmod(k, N)$, $i'=\lfloor k-1/N \rfloor$, $j'=\bmod(k-1, N)$, and
\begin{equation}
\mathrm{g}(x, y)=(x+A*y) \bmod{256}.
\label{eq:difussion}
\end{equation}

\item \textit{Step (d) repetition}: repeat \textit{Step (a)} through \textit{Step (c)} $T-1$ times.

\item \textit{Step (e) final mask substitution}: generate the two half parts of cipher-image as follows: do
\begin{equation}
\bm{r}=\bm{V}_{T}\oplus \bm{L}_{T}
\label{eq:maskleft}
\end{equation}
and
\begin{equation}
\bm{l}=\bm{V}_{T+1}\oplus \bm{R}_{T},
\label{eq:maskright}
\end{equation}
where the exclusive or operation between two matrixes is calculated element-wise, the same hereinafter.
\end{itemize}

\item \textit{The decryption procedure} is similar to the encryption process except the following simple modifications: 1) the \textit{Step (e)} is performed first; 2) the different rounds of encryption are exerted in a reverse order.
\end{itemize}

\begin{figure}[!htb]
\centering
\centering
\includegraphics[width=\imagewidth]{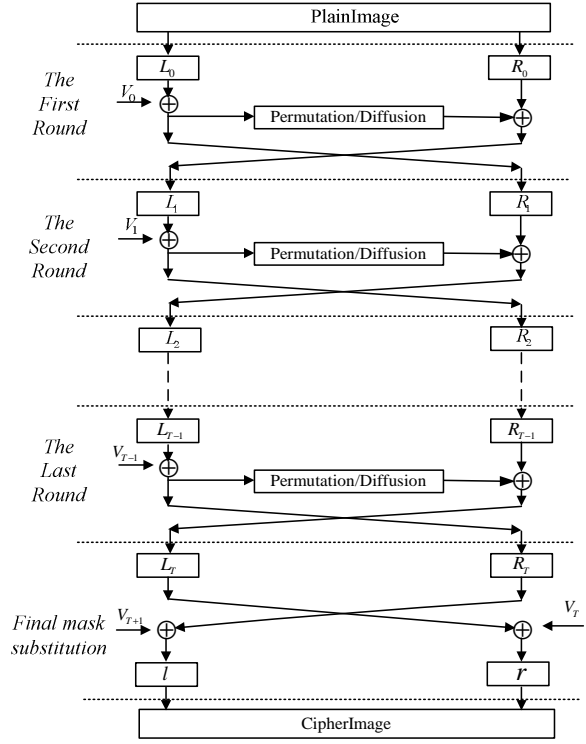}
\caption{Schematic structure of IEAS.}
\label{fig:originstructure}
\end{figure}

\section{Differential cryptanalysis}
\label{sec:cryptanalysis}

Task of differential cryptanalysis is to get information of (equivalent) secret key of an encryption algorithm by observing how differences in an input can affect the resultant ones at the output. Generally, the difference is defined with respect to exclusive or (XOR) operation. In the following, some properties of IEAS are introduced first, which works as basis for differential attack on IEAS under different round numbers.

\subsection{Some properties of IEAS}

\begin{property}
Given two matrix entries $(i_1, j_1)$ and $(i_2 ,j_2)$ in $\bm{R}_{l}$, and let $(\tilde{i}_1, \tilde{j}_1)$ and $(\tilde{i}_2, \tilde{j}_2)$ denote the corresponding locations in $\widetilde{\bm{R}}_{l}$. If the two original entries satisfy
\begin{equation}
\gcd (\Delta,N) =1,
\label{eq:condition}
\end{equation}
one has
\begin{equation*}
\bm{C}_l=
\left(
\begin{matrix}
s  & u \\
v  & t
\end{matrix}\right),
\end{equation*}
where
\begin{align*}
\left(
\begin{matrix}
s \\
u \\
v \\
t
\end{matrix}
\right) & =
\left(
\begin{matrix}
\Delta^{-1}(\tilde{i}_1j_2-\tilde{i}_2j_1) \\
\Delta^{-1}(\tilde{i}_2i_1-\tilde{i}_1i_2) \\
\Delta^{-1}(\tilde{j}_1j_2-\tilde{j}_2j_1) \\
\Delta^{-1}(\tilde{j}_2i_1-\tilde{j}_1i_2)
\end{matrix}
\right)
\bmod N,
\end{align*}
$\Delta=i_1j_2-i_2j_1$, and $\Delta\cdot\Delta^{-1}=1 \bmod N$.
\label{Pro:calculate}
\end{property}
\begin{proof}
Obviously, $(i_1, j_1)$, $(i_2,j_2)$, $(\tilde{i}_1, \tilde{j}_1)$, and $(\tilde{i}_2, \tilde{j}_2)$ satisfy
\begin{align*}
\left(
\begin{matrix}
si_1   + uj_1  \\
si_2 + uj_2
\end{matrix}
\right)\bmod N & =
\left(
\begin{matrix}
\tilde{i}_1 \\
\tilde{i}_2
\end{matrix}
\right),
\end{align*}
which means
\begin{equation*}
\left(
\begin{matrix}
i_1 & j_1 \\
i_2 & j_2
\end{matrix}
\right) \cdot {s \choose u} = {\tilde{i}_1+K_1N \choose \tilde{i}_2+K_2N},
\end{equation*}
where $K_1, K_2\in \mathbb{Z}$.

Use the Gaussian elimination method, one can get
\begin{equation*}
\left(
\begin{matrix}
i_1 & j_1 \\
0 & i_1j_2-i_2j_1
\end{matrix}
\right) \cdot {s \choose u} = {\tilde{i}_1+K_1N \choose i_1\tilde{i}_2-i_2\tilde{i}_1+N(K_2i_1-K_1i_2)}.
\end{equation*}
According to the Cramer's rule, the above equation have one and only one solution when $\gcd (\Delta,N) =1$.
Thus,
\begin{align*}
s &=\Delta^{-1}(\tilde{i}_1j_2-\tilde{i}_2j_1) \bmod N, \\
u &=\Delta^{-1}(\tilde{i}_2i_1-\tilde{i}_1i_2) \bmod N.
\end{align*}
The value of $v, t$ can be obtained similarly.
\end{proof}

\begin{property}
If $2^n$ $(1\le n\le 7)$ divides variable $A$ in Eq.~(\ref{eq:difussion}), then the substitution function $\mathrm{g}(x, y)$ has no influence on the $n$ least significant bits of $x$, i.e., Eq.~(\ref{eq:substitution}) becomes
\begin{equation*}
\bm{L}_{l,k}(i,j)=\bm{R}_{l-1,k}(i,j)\oplus \widetilde{\bm{R}}_{l,k}(i,j),
\end{equation*}
where $k\in \{1, \cdots, n\}$, $\bm{L}_{l,k}$, $\bm{R}_{l-1,k}$ and $\widetilde{\bm{R}}_{l,k}$ are the $k$-th least significant bit plane of $\bm{L}_{l}$, $\bm{R}_{l-1,}$ and $\widetilde{\bm{R}}_{l}$, respectively.
\end{property}
\begin{proof}
This property can be easily proved by calculating
\begin{IEEEeqnarray*}{rCl}
\mathrm{g}(x, y) & = & x + A \cdot \sum_{i=0}^{7} y_i 2^{i}  \bmod {256}\\
                 & = & x + (A/2^n) \cdot \sum_{i=n}^{7}y_i 2^{i} \bmod{256}.
\end{IEEEeqnarray*}
\end{proof}

Let $\bm{L}'_{l}$, $\bm{R}'_{l}(\bm{P}_{l-1})$ and $\bm{R}'_{l-1}$ denote differential of two versions of $\bm{L}_{l}$, $\bm{R}_{l}(\bm{P}_{l-1})$ and $\bm{R}_{l-1}$, respectively. Observe the structure of intermediate data under different rounds shown in Fig.~\ref{fig:structure}, we can get the following property.

\begin{property}
If $2^n$ $(1\le n\le 7)$ divides variable $A$ in Eq.~(\ref{eq:difussion}), one has
\begin{equation*}
\left\{\,
\begin{IEEEeqnarraybox}[][c]{rCl}
\IEEEstrut
\bm{R}_l'      & = & \bm{L}_{l-1}', \\
\bm{L}'_{l, k} & = & \bm{R}'_{l-1, k}\oplus \bm{R}'_{l, k}(\bm{P}_{l-1}),
\IEEEstrut
\end{IEEEeqnarraybox}
\right.
\end{equation*}
where $k\in \{1, \cdots, n\}$, $\bm{L}'_{l, k}$, $\bm{R}'_{l-1, k}$ and  $\bm{R}'_{l, k}(\bm{P}_{l-1})$ are the $k$-th least significant bit plane of $\bm{L}'_{l}$, $\bm{R}'_{l-1}$ and $\bm{R}'_{l}(\bm{P}_{l-1})$, respectively.
\label{Pro:differentialstructure}
\end{property}
\begin{proof}
This property can be easily proved with mathematical induction on $l$ $(1\leq l\leq T)$. When $l=1$,
\begin{IEEEeqnarray*}{rCl}
\bm{R}_1'     & = & \bm{R}_1\oplus \bm{R}_1^*  \\
              & = & (\bm{L}_0\oplus \bm{V}_0)\oplus (\bm{L}_0^*\oplus \bm{V}_0) \\
              & = & \bm{L}_0', \\
\bm{L}'_{1,k} & = & (\bm{R}_{0,k} \oplus \bm{R}_{1,k}(\bm{P}_{0})) \oplus (\bm{R}^*_{0,k} \oplus \bm{R}^*_{1,k}(\bm{P}_{0}))\\
              & = & \bm{R}'_{0,k}\oplus \bm{R}'_{1,k}(\bm{P}_{0}).
\end{IEEEeqnarray*}
So, the property holds for $l=1$.
Assume that the property is true for $l=n$ $(n<T)$, we prove the case for $l=n+1$.
\begin{IEEEeqnarray*}{rCl}
\bm{R}_{n+1}'   & = & (\bm{L}_n\oplus \bm{V}_n)\oplus (\bm{L}_n^*\oplus \bm{V}_n) \\
                & = & \bm{L}_n', \\
\bm{L}'_{n+1,k} & = & (\bm{R}_{n,k} \oplus \bm{R}_{n+1,k}(\bm{P}_{n})) \oplus (\bm{R}^*_{n,k} \oplus \bm{R}^*_{n+1,k}(\bm{P}_{n}))\\
                & = & \bm{R}'_{n,k}\oplus \bm{R} '_{n+1,k}(\bm{P}_{n}).
\end{IEEEeqnarray*}
This completes the mathematical induction.
\end{proof}

\begin{figure}[!htb]
\centering
\begin{minipage}[t]{\imagewidth}
\centering
\includegraphics[width=\imagewidth]{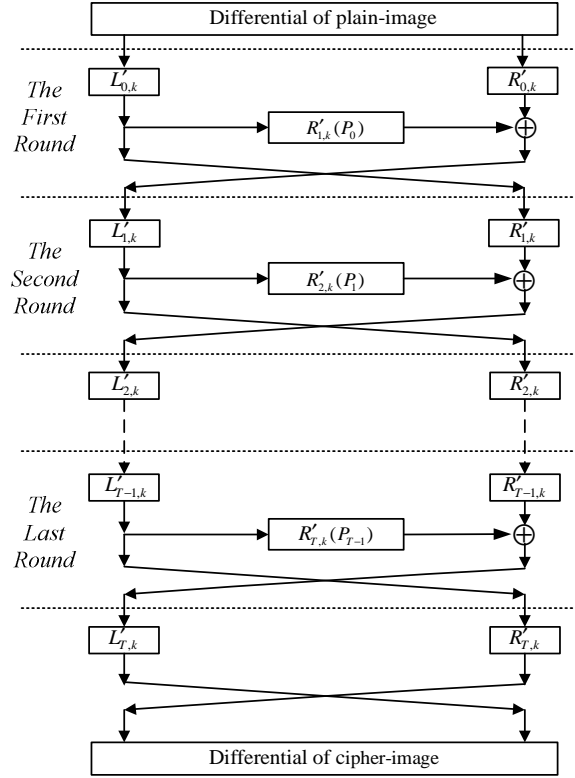}
\end{minipage}
\caption{Schematic structure of differential of intermediate data under different rounds.}
\label{fig:structure}
\end{figure}

\subsection{Breaking IEAS when the number of iteration round is equal to one}

Given two known/chosen plain-images, $[\bm L_{0},\bm R_{0}]$ and $[\bm L_{0}^*,\bm R_{0}^*]$, and the corresponding cipher-images, $[\bm{l},\bm{r}]$ and $[\bm{l}^*,\bm{r}^*]$, one has
\begin{equation*}
\left\{\,
\begin{IEEEeqnarraybox}[][c]{rCl}
\IEEEstrut
\bm{L}_0' &= &\bm{L}_0 \oplus \bm{L}_0^*,\\
\bm{R}_0' &=& \bm{R}_0 \oplus \bm{R}_0^*,
\IEEEstrut
\end{IEEEeqnarraybox}
\right.
\end{equation*}
and
\begin{equation*}
\left\{\,
\begin{IEEEeqnarraybox}[][c]{rCl}
\IEEEstrut
\bm{L}_1' &=& \bm{r} \oplus \bm{r}^*,\\
\bm{R}_1' &=& \bm{l} \oplus \bm{l}^*.
\IEEEstrut
\end{IEEEeqnarraybox}
\right.
\end{equation*}

From Property~\ref{Pro:differentialstructure}, one can get
\begin{equation}
\left\{\,
\begin{IEEEeqnarraybox}[][c]{rCl}
\IEEEstrut
\bm{R}_1'      & = & \bm{L}_{0}', \\
\bm{R}'_{1, k}(\bm{P}_{0}) & = &\bm{L}'_{1, k}  \oplus \bm{R}'_{0, k},
\IEEEstrut
\end{IEEEeqnarraybox}
\right.
\label{eq:informationP0}
\end{equation}
where $k\in \{1, \cdots, n\}$, $2^n$ $(1\le n\le 7)$ divides the parameter $A$ in Eq.~(\ref{eq:difussion}). Comparing $\{\bm{R}_{1,k}'\}_{k=1}^n$ and $\{\bm{R}'_{1, k}(\bm{P}_{0})\}_{k=1}^n$, one may find two pairs of entries in $\bm{R}_1' $ and $\bm{R}'_{1}(\bm{P}_{0})$ whose locations satisfying condition~(\ref{eq:condition}). Then, the transformation matrix $\bm{C}_0$, generating the associated permutation matrix $\bm{P}_{0}$, can be solved according to Property~\ref{Pro:calculate}. In case the search of the required entries failed, one can resort to observing more known plain-images or constructing special differential images from more chosen plain-images \cite{Lcq:Optimal:SP11}. As shown in \cite{Lcq:Optimal:SP11}, $\lceil 2\log_2(N)\rceil$ chosen binary plain-images are enough to break any position permutation-only encryption algorithm exerting on binary plain-images of size $N\times N$. Due to similarity, we do not mention the problem about determining permutation matrix with more known/chosen plain-images in the remainder of this paper. Once $\bm{C}_0$ is determined, the associated matrix $\bm{P}_{0}$ can be obtained from it easily.

Referring to Eq.~(\ref{eq:substitutionxor}) and Eq.~(\ref{eq:maskright}), one can get
\begin{equation}
\bm{V}_2 \oplus \bm{V}_0=\bm{L}_0 \oplus \bm{l}.
\label{eq:condition1R1}
\end{equation}
Combining Eq.~(\ref{eq:maskleft}) and Eq.~(\ref{eq:substitution}) yields
\begin{equation}
\bm{r}_k =\bm{V}_{1,k} \oplus \bm{R}_{0,k} \oplus \bm{R}_{1,k}(\bm{P}_{0}), \label{eq:expandR1}
\end{equation}
where $\bm{r}_k$ is the $k$-th least significant bit plane of $\bm{r}$. As the exclusive or operation is linear with respect to position permutation, one can get
\begin{equation*}
\bm{R}_{1,k}(\bm{P}_{0}) = \bm{V}_{0,k}(\bm{P}_{0}) \oplus \bm{L}_{0,k}(\bm{P}_0)
\end{equation*}
from Eq.~(\ref{eq:substitutionxor}).
Substitute $\bm{R}_{1,k}(\bm{P}_{0})$ obtained in the above equation into Eq.~(\ref{eq:expandR1}), one can further get
\begin{equation}
\bm{V}_{1,k} \oplus \bm{V}_{0,k}(\bm{P}_{0})=\bm{r}_k \oplus \bm{R}_{0,k}  \oplus \bm{L}_{0,k}(\bm{P}_{0}).
\label{eq:equivalentkeyR1}
\end{equation}

Since neither of Eq.~(\ref{eq:condition1R1}) and Eq.~(\ref{eq:equivalentkeyR1}) has any special requirement on the pair of plain-image and corresponding cipher-image, some parts of any other cipher-image encrypted with the same secret key, $[\bm{l}^{\star}, \bm{r}^{\star}]$, can be recovered by calculating \begin{equation*}
\left\{\,
\begin{IEEEeqnarraybox}[][c]{rCl}
\IEEEstrut
\bm{L}^{\star}_0       &=& \bm{l}^{\star}  \oplus \bm{M}_{1}  , \\
\bm{R}^{\star}_{0,k}   &=& \bm{r}^{\star}_k \oplus \bm{L}^{\star}_{0,k}(\bm{P}_{0})
\oplus \bm{N}_{1,k}  ,
\IEEEstrut
\end{IEEEeqnarraybox}
\right.
\end{equation*}
where
\begin{equation*}
\left\{\,
\begin{IEEEeqnarraybox}[][c]{rCl}
\IEEEstrut
\bm{M}_{1}      &=& \bm{L}_0 \oplus \bm{l}, \\
\bm{N}_{1,k}    &=&\bm{r}_k \oplus \bm{R}_{0,k}  \oplus \bm{L}_{0,k}(\bm{P}_{0}).
\IEEEstrut
\end{IEEEeqnarraybox}
\right.
\end{equation*}
Now, one can see that $\bm{M}_{1}$, $\{\bm{N}_{1,k}\}_{k=1}^n$ and $\bm{P}_{0}$ can work together to recover the whole left half part of $\bm{l}^{\star}$, and the $n$ least significant bit planes of the right part of $\bm{r}^{\star}$, $\{\bm{R}^{\star}_{0,k}\}_{k=1}^n$.

To verify the above analysis, some experiments on some plain-images of size $256\times 512$ are made. With secret key $K_0=1234567/(2^{32}-1)$, $T=1$ and parameter $A=64$, two known plain-images, cropped version of two standard images, ``Lenna" and ``Baboon", and the corresponding cipher-images are shown in Figs.~\ref{figure:differentialattackR1}a), b), d), e), respectively. The obtained information about the secret key is used to decrypt another cipher-image shown in Fig.~\ref{figure:differentialattackR1}c), and result is shown in Fig.~\ref{figure:differentialattackR1}f). The whole left half part and the 6 least significant bit planes of the right half part of the recovered image shown in Fig.~\ref{figure:differentialattackR1}f) are identical with counterpart of the corresponding plain-image, which agree with the expected result well.

\begin{figure}[!htb]
\centering
\begin{minipage}{\figwidth}
\includegraphics[width=\textwidth]{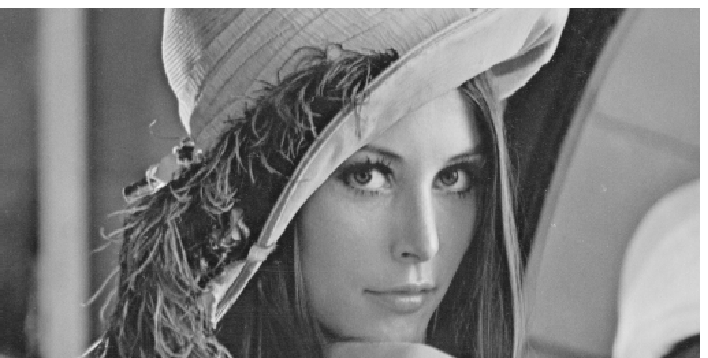}
a)
\end{minipage}
\begin{minipage}{\figwidth}
\includegraphics[width=\textwidth]{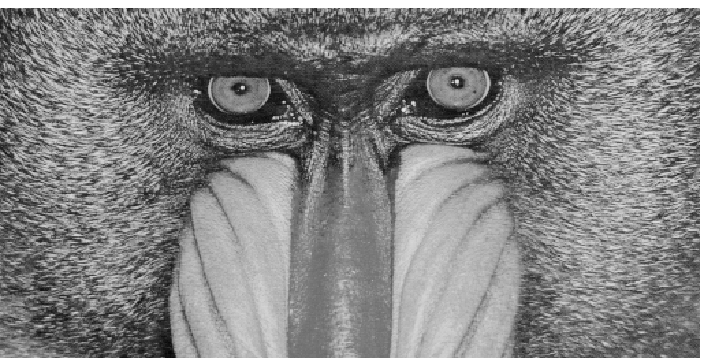}
b)
\end{minipage}
\begin{minipage}{\figwidth}
\includegraphics[width=\textwidth]{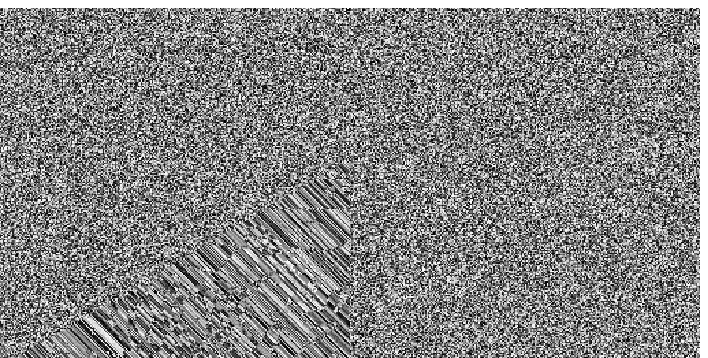}
c)
\end{minipage}
\\
\begin{minipage}{\figwidth}
\includegraphics[width=\textwidth]{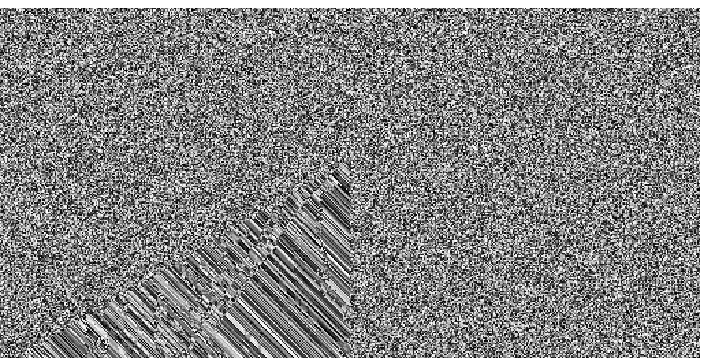}
d)
\end{minipage}
\begin{minipage}{\figwidth}
\includegraphics[width=\textwidth]{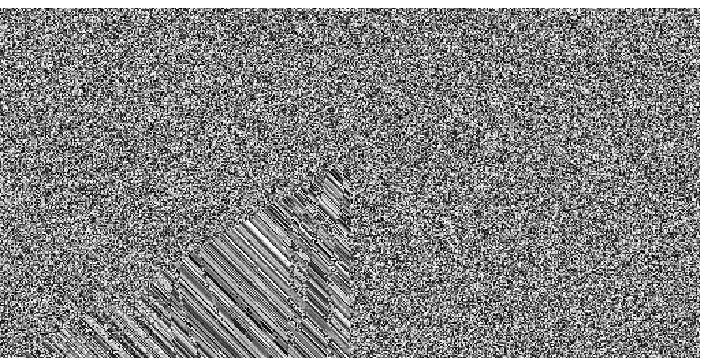}
e)
\end{minipage}
\begin{minipage}{\figwidth}
\includegraphics[width=\textwidth]{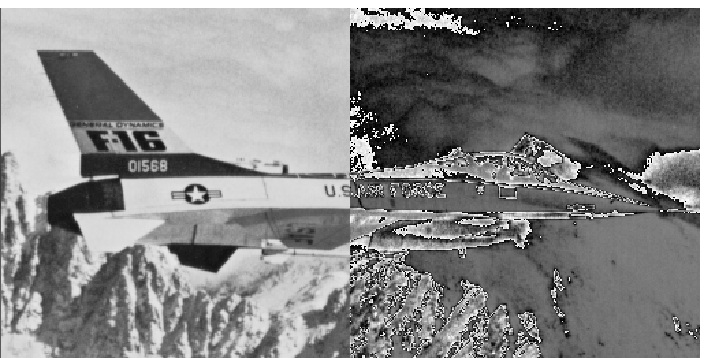}
f)
\end{minipage}
\caption{Differential attack on IEAS when $T=1$:
a) the first known plain-image;
b) the second known plain-image;
c) cipher-image of plain-image ``Airplane";
d) cipher-image of Fig.~\ref{figure:differentialattackR1}a);
e) cipher-image of Fig.~\ref{figure:differentialattackR1}b);
f) the recovered plain-image of Fig.~\ref{figure:differentialattackR1}c).}
\label{figure:differentialattackR1}
\end{figure}

\subsection{Breaking IEAS when the number of iteration round is equal to two}

In this case, the differential of ciphertext is
\begin{equation}
\left\{\,
\begin{IEEEeqnarraybox}[][c]{rCl}
\IEEEstrut
\bm{L}_2' &=& \bm{r} \oplus \bm{r}^*,\\
\bm{R}_2' &=& \bm{l} \oplus \bm{l}^*.
\IEEEstrut
\end{IEEEeqnarraybox}
\right.
\label{eq:ciherdifferentialR2}
\end{equation}
From Property 3, one has
\begin{equation*}
\left\{\,
\begin{IEEEeqnarraybox}[][c]{rCl}
\IEEEstrut
\bm{R}_1'                   & = & \bm{L}_{0}', \\
\bm{R}'_{1, k}(\bm{P}_{0})  & = & \bm{R}'_{2,k} \oplus \bm{R}'_{0,k},
\IEEEstrut
\end{IEEEeqnarraybox}
\right.
\end{equation*}
where $\bm{R}'_{2,k}$ is the $k$-th least significant bit plane of $\bm{R}'_{2}$. Then, the transformation matrix $\bm{C}_{0}$, generating the associated permutation matrix $\bm{P}_{0}$, can be recovered by comparing $\{\bm{R}_{1,k}'\}_{k=1}^n$ and $\{\bm{R}'_{1, k}(\bm{P}_{0})\}_{k=1}^n$.

Still from Property~\ref{Pro:differentialstructure}, one has
\begin{IEEEeqnarray*}{rCl}
\bm{R}'_{2, k}(\bm{P}_{1})  &=& \bm{L}'_{2, k}  \oplus \bm{R}'_{1, k} \\
                             &=& \bm{L}'_{2,k} \oplus \bm{L}'_{0,k}.
\end{IEEEeqnarray*}
Similarly, one can get the transform matrix $\bm{C}_{1}$, then permutation matrix $\bm{P}_{1}$, by comparing $\{\bm{R}_{2,k}'\}_{k=1}^n$ and $\{\bm{R}'_{2, k}(\bm{P}_{1})\}_{k=1}^n$.

Referring to Eq.~(\ref{eq:maskright}), one has
\begin{IEEEeqnarray}{rCl}
\bm{l}_k  &=& \bm{V}_{3,k} \oplus \bm{R}_{2,k} \label{eq:differentialR2},
\end{IEEEeqnarray}
where $\bm{l}_k$ is the $k$-th least significant bit plane of $\bm{l}$.
Combining Eq.~(\ref{eq:substitutionxor}), Eq.~(\ref{eq:substitution}) and Eq.~(\ref{eq:maskleft}) yields
\begin{IEEEeqnarray}{rCl}
\bm{r}_k  &=& \bm{V}_{2,k} \oplus \bm{L}_{2,k}  \nonumber\\
          &=& \bm{V}_{2,k} \oplus \bm{R}_{1,k} \oplus \bm{R}_{2,k}(\bm{P}_{1}) \nonumber\\
          &=& \bm{V}_{2,k} \oplus \bm{V}_{0,k} \oplus \bm{L}_{0,k} \oplus \bm{R}_{2,k}(\bm{P}_{1}).\nonumber
\end{IEEEeqnarray}
Substitute $\bm{R}_{2,k}$ obtained in Eq.~(\ref{eq:differentialR2}) into the above equation, one has
\begin{equation}
\bm{V}_{3,k}(\bm{P}_1) \oplus \bm{V}_{2,k} \oplus \bm{V}_{0,k} = \bm{l}_{k}(\bm{P}_1) \oplus \bm{r}_{k}
\oplus \bm{L}_{0,k}.
\label{eq:resultR21}
\end{equation}
Combine Eq.~(\ref{eq:substitutionxor}) and Property~2, one can get
\begin{IEEEeqnarray}{rCl}
\bm{R}_{2,k}  &=& \bm{V}_{1,k} \oplus \bm{L}_{1,k} \nonumber \\
          &=& \bm{V}_{1,k} \oplus \bm{R}_{0,k} \oplus \bm{R}_{1,k}(\bm{P}_{0}) \nonumber \\
          &=& \bm{V}_{1,k} \oplus \bm{R}_{0,k} \oplus \bm{V}_{0,k}(\bm{P}_{0}) \oplus \bm{L}_{0,k}(\bm{P}_{0}),
          \label{eq:resulttoR22}
\end{IEEEeqnarray}
then Eq.~(\ref{eq:differentialR2}) becomes
\begin{equation}
\bm{V}_{3,k} \oplus \bm{V}_{1,k} \oplus \bm{V}_{0,k}(\bm{P}_0) = \bm{l}_{k} \oplus \bm{L}_{0,k}(\bm{P}_0)  \oplus \bm{R}_{0,k}.
\label{eq:resultR22}
\end{equation}

Since both Eq.~(\ref{eq:resultR21}) and Eq.~(\ref{eq:resultR22}) always hold for any pair of plain-image and cipher-image encrypted with
the same secret key, it can be easily verified that
\begin{equation*}
\left\{
\begin{IEEEeqnarraybox}[][c]{rCl}
\IEEEstrut
\bm{L}^{\star}_{0,k} &=& \bm{l}^{\star}_k(\bm{P}_{1}) \oplus \bm{r}^{\star}_k \oplus  \bm{M}_{2,k} , \\
\bm{R}^{\star}_{0,k} &=& \bm{l}^{\star}_k \oplus \bm{L}^{\star}_{0,k}(\bm{P}_{0}) \oplus \bm{N}_{2,k},
\IEEEstrut
\end{IEEEeqnarraybox}
\right.
\end{equation*}
where
\begin{equation*}
\left\{\,
\begin{IEEEeqnarraybox}[][c]{rCl}
\IEEEstrut
\bm{M}_{2,k}  &=&  \bm{l}_k(\bm{P}_{1})\oplus \bm{r}_k \oplus \bm{L}_{0,k} , \\
\bm{N}_{2,k}  &=& \bm{l}_k \oplus \bm{L}_{0,k}(\bm{P}_{0})  \oplus \bm{R}_{0,k}.
\IEEEstrut
\end{IEEEeqnarraybox}
\right.
\end{equation*}
The above equations mean that $\bm{M}_{2,k}$, $\bm{N}_{2,k}$ and $\{\bm{P}_{l}\}_{l=0}^1$ can work together to recover the $k$-th least significant bit plane of any other cipher-image encrypted with the same secret key, $[\bm{L}^{\star}_{0,k}, \bm{R}^{\star}_{0,k}]$, for $k=1\sim n$.

To verify the above analysis, some experiments are made. With secret key $K_0=1234567/(2^{32}-1)$, $T=2$ and parameter $A=64$, encryption results of the two known-images shown in Figs.~\ref{figure:differentialattackR1}a) and b) are shown in Figs.~\ref{figure:differentialattackR2}a) and b), respectively. The information about equivalent secret key obtained from the two pairs of plain-images and cipher-images is used decrypt another cipher-image shown in Fig.~\ref{figure:differentialattackR2}c) and the result is shown in Fig.~\ref{figure:differentialattackR2}d). It is counted that the $6$ least significant bit planes of the image shown in Fig.~\ref{figure:differentialattackR2}d) are identical with the counterparts of the corresponding plain-image.

\begin{figure}[!htb]
\centering
\begin{minipage}{\figwidth}
\includegraphics[width=\textwidth]{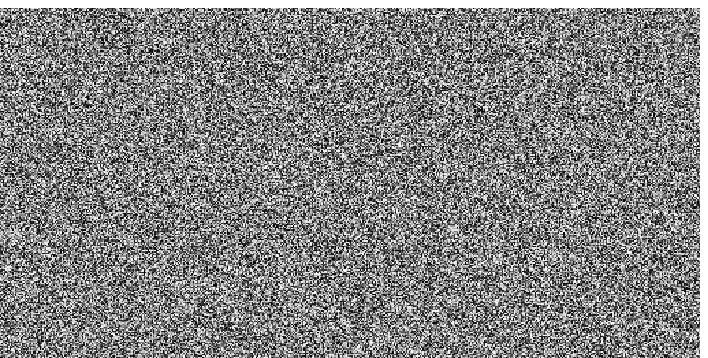}
a)
\end{minipage}
\begin{minipage}{\figwidth}
\includegraphics[width=\textwidth]{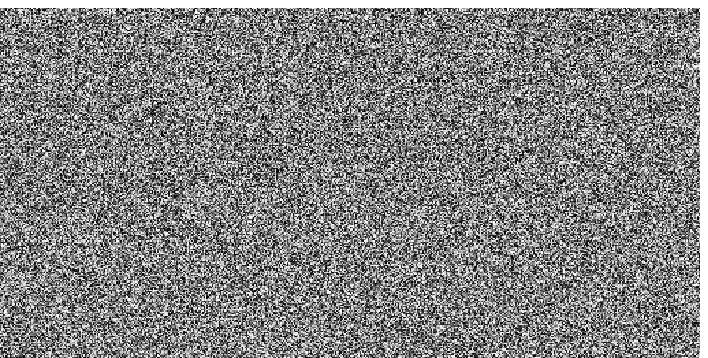}
b)
\end{minipage}\\
\begin{minipage}{\figwidth}
\includegraphics[width=\textwidth]{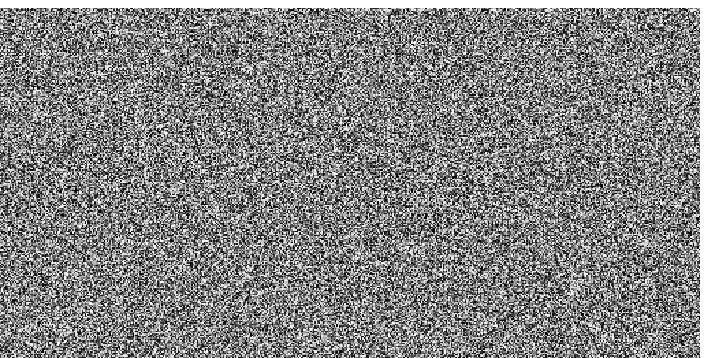}
c)
\end{minipage}
\begin{minipage}{\figwidth}
\includegraphics[width=\textwidth]{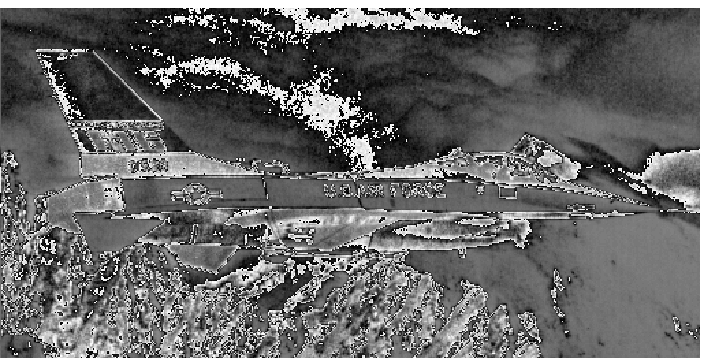}
d)
\end{minipage}
\caption{Differential attack on IEAS when $T=2$:
a) cipher-image of Fig.~\ref{figure:differentialattackR1}a);
b) cipher-image of Fig.~\ref{figure:differentialattackR1}b);
c) cipher-image of plain-image ``Airplane";
d) the recovered plain-image of Fig.~\ref{figure:differentialattackR2}c).}
\label{figure:differentialattackR2}
\end{figure}

\subsection{Breaking IEAS when the round number is equal to three}

In this sub-section we discuss how to break the version of IEAS of three rounds with no less than three chosen plain-images.

In this case, the differential of ciphertext is
\begin{equation*}
\left\{\,
\begin{IEEEeqnarraybox}[][c]{rCl}
\IEEEstrut
\bm{L}_3' &= & \bm{r} \oplus \bm{r}^*,\\
\bm{R}_3' &= &\bm{l} \oplus \bm{l}^*.
\IEEEstrut
\end{IEEEeqnarraybox}
\right.
\end{equation*}
According to Property~\ref{Pro:differentialstructure}, one has
\begin{IEEEeqnarray}{rCl}
\bm{L}_{3,k}'  &=& \bm{R}'_{3, k}(\bm{P}_{2}) \oplus \bm{R}'_{2,k}                                   \nonumber\\
               &=& \bm{R}'_{3, k}(\bm{P}_{2}) \oplus \bm{R}'_{0,k} \oplus \bm{R}'_{1, k}(\bm{P}_{0}) \nonumber\\
               &=& \bm{R}'_{3, k}(\bm{P}_{2}) \oplus \bm{R}'_{0,k} \oplus \bm{L}'_{0, k}(\bm{P}_{0}).
                \label{eq:dectyptP2R3}
\end{IEEEeqnarray}
If $\bm{L}'_{0,k}$ is chosen as a binary matrix of fixed value, which makes
\begin{equation}
\bm{L}'_{0, k}(\bm{P}_{0})\equiv\bm{L}'_{0, k},
\label{eq:cancelpermute}
\end{equation}
$\bm{R}'_{3, k}(\bm{P}_{2})$ can be obtained from Eq.~(\ref{eq:dectyptP2R3}). With the same method
mentioned above, $\bm{P}_{2}$ can be recovered by comparing $\{\bm{R}_{3,k}'\}_{k=1}^n$ and $\{\bm{R}'_{3, k}(\bm{P}_{2})\}_{k=1}^n$.

As for the differential image satisfying Eq.~(\ref{eq:cancelpermute}), one also has
\begin{IEEEeqnarray}{rCl}
\bm{R}'_{2,k}   &=&  \bm{R}'_{3,k}(\bm{P}_{2}) \oplus \bm{L}_{3,k}' \nonumber \\
                &=&  \bm{R}'_{0,k} \oplus \bm{L}'_{0, k}(\bm{P}_{0}) \nonumber \\
                &=&  \bm{R}'_{0,k} \oplus \bm{L}'_{0, k}.     \label{eq:R2kR3}
\end{IEEEeqnarray}
Note that
\begin{IEEEeqnarray}{rCl}
\bm{R}_{3,k}'   &=& \bm{L}'_{2,k}                                   \nonumber \\
                &=& \bm{R}'_{1,k} \oplus \bm{R}'_{2, k}(\bm{P}_{1}). \nonumber
\end{IEEEeqnarray}
Substitute Eq.~(\ref{eq:R2kR3}) into the above equation, one can get
\begin{IEEEeqnarray*}{rCl}
\bm{R}_{3,k}'  &=& \bm{R}'_{1,k} \oplus \bm{R}'_{0,k}(\bm{P}_{1}) \oplus \bm{L}'_{0, k}(\bm{P}_{1}) \\
               &=& \bm{L}'_{0,k} \oplus \bm{R}'_{0,k}(\bm{P}_{1}) \oplus \bm{L}'_{0, k}\\
               &=&\bm{R}'_{0,k}(\bm{P}_{1}) .
\end{IEEEeqnarray*}
Then, $\bm{R}'_{0, k}(\bm{P}_{1})$ can be obtained from the above equation, and $\bm{P}_{1}$ can be recovered by comparing $\{\bm{R}_{0,k}'\}_{k=1}^n$ and $\{\bm{R}'_{0, k}(\bm{P}_{1})\}_{k=1}^n$.

Once $\bm{P}_{2}$ is recovered, $\bm{L}'_{0, k}(\bm{P}_{0})$ can be obtained from Eq.~(\ref{eq:dectyptP2R3}). Then, $\bm{P}_{0}$ can be recovered by comparing $\{\bm{L}_{0,k}'\}_{k=1}^n$ and $\{\bm{L}'_{0, k}(\bm{P}_{0})\}_{k=1}^n$.
As mentioned before, one and even more pairs of plaintext and the corresponding ciphertexts are required to
find two pairs of entries in $\bm{L}_{0,k}'$ and $\bm{L}'_{0, k}(\bm{P}_{0})$ whose locations satisfying condition~(\ref{eq:condition}).

From Eq.~(\ref{eq:maskright}), one has
\begin{equation}
\bm{l}_k  = \bm{V}_{4,k} \oplus \bm{R}_{3,k},\label{eq:expressLKR3}
\end{equation}
where $\bm{R}_{3,k}$ is the $k$-th least significant bit plane of $\bm{R}_{3}$.
Combine Eq.~(\ref{eq:substitutionxor}), Eq.~(\ref{eq:substitution}) and Eq.~(\ref{eq:maskleft}), one can get
\begin{IEEEeqnarray}{rCl}
\bm{r}_k  &=& \bm{V}_{3,k} \oplus \bm{L}_{3,k} \nonumber  \\
          &=& \bm{V}_{3,k} \oplus \bm{R}_{2,k} \oplus \bm{R}_{3,k}(\bm{P}_{2}) \nonumber \\
          &=& \bm{V}_{3,k} \oplus \bm{V}_{1,k} \oplus \bm{L}_{1,k} \oplus \bm{R}_{3,k}(\bm{P}_{2}) \nonumber \\
          &=& \bm{V}_{3,k} \oplus \bm{V}_{1,k} \oplus \bm{R}_{0,k} \oplus \bm{R}_{1,k}(\bm{P}_{0})\oplus \bm{R}_{3,k}(\bm{P}_{2}) \nonumber \\
          &=& \bm{V}_{3,k} \oplus \bm{V}_{1,k} \oplus \bm{R}_{0,k} \oplus \bm{V}_{0,k}(\bm{P}_{0})\oplus \bm{L}_{0,k}(\bm{P}_{0})\oplus \bm{R}_{3,k}(\bm{P}_{2}).\nonumber
\end{IEEEeqnarray}
Substitute $\bm{R}_{3,k}$ obtained in Eq.~(\ref{eq:expressLKR3}) into the above equation and
get
\begin{equation}
 \bm{V}_{4,k}(\bm{P}_{2}) \oplus  \bm{V}_{3,k} \oplus \bm{V}_{1,k} \oplus\bm{V}_{0,k}(\bm{P}_{0})  =  \bm{l}_{k}(\bm{P}_{2}) \oplus\bm{r}_k \oplus \bm{R}_{0,k} \oplus  \bm{L}_{0,k}(\bm{P}_{0}). \label{eq:equivalentkey1R3}
\end{equation}
Referring to Eq.~(\ref{eq:substitutionxor}) and Eq.~(\ref{eq:resulttoR22}), one has
\begin{IEEEeqnarray}{rCl}
\bm{R}_{3,k} &=& \bm{V}_{2,k} \oplus \bm{L}_{2,k} \nonumber \\
             &=& \bm{V}_{2,k} \oplus \bm{R}_{1,k} \oplus \bm{R}_{2}(\bm{P}_{1}) \nonumber \\
             &=& \bm{V}_{2,k} \oplus \bm{V}_{0,k} \oplus \bm{L}_{0,k} \oplus \bm{V}_{1,k}(\bm{P}_1) \oplus \bm{R}_{0,k}(\bm{P}_{1})\oplus \bm{V}_{0,k}(\bm{P}_0\bm{P}_{1})\oplus \bm{L}_{0,k}(\bm{P}_0\bm{P}_{1}),
 \label{eq:expressR2KR3}
\end{IEEEeqnarray}
then Eq.~(\ref{eq:expressLKR3}) can be rewritten as
\begin{equation*}
\bm{l}_k   = \bm{V}_{4,k}\oplus \bm{V}_{2,k} \oplus \bm{V}_{0,k} \oplus \bm{L}_{0,k} \oplus \bm{V}_{1,k}(\bm{P}_1) \oplus \bm{R}_{0,k}(\bm{P}_{1})\oplus \bm{V}_{0,k}(\bm{P}_0\bm{P}_{1})\oplus \bm{L}_{0,k}(\bm{P}_0\bm{P}_{1}).
\end{equation*}
Substitute $\bm{L}_{0,k}(\bm{P}_{0})$ obtained in Eq.~(\ref{eq:equivalentkey1R3}) into the above equation, and get
\begin{equation}
\bm{V}_{4,k}(\bm{P}_{2}\bm{P}_{1}) \oplus \bm{V}_{3,k}(\bm{P}_1) \oplus \bm{V}_{4,k} \oplus \bm{V}_{2,k} \oplus \bm{V}_{0,k}  = \bm{l}_{k}(\bm{P}_{2}\bm{P}_{1}) \oplus \bm{r}_{k}(\bm{P}_{1}) \oplus \bm{l}_k \oplus  \bm{L}_{0,k}.\label{eq:equivalentkey2R3}
\end{equation}
Since both Eq.~(\ref{eq:equivalentkey1R3}) and Eq.~(\ref{eq:equivalentkey2R3}) hold for any pair of plain-image and its corresponding cipher-image, it can be easily verified that
\begin{equation*}
\left\{\,
\begin{IEEEeqnarraybox}[][c]{rCl}
\IEEEstrut
\bm{L}^{\star}_{0,k} &=& \bm{l}^{\star}_k(\bm{P_2}\bm{P_1}) \oplus  \bm{r}^{\star}_k(\bm{P_1}) \oplus \bm{l}^{\star}_k \oplus \bm{M}_{3,k}, \\
\bm{R}^{\star}_{0,k} &=&\bm{l}^{\star}_k(\bm{P_2}) \oplus  \bm{r}^{\star}_k   \oplus \bm{L}^{\star}_{0,k}(\bm{P}_{0}) \oplus \bm{N}_{3,k},
\IEEEstrut
\end{IEEEeqnarraybox}
\right.
\end{equation*}
where
\begin{equation*}
\left\{\,
\begin{IEEEeqnarraybox}[][c]{rCl}
\IEEEstrut
\bm{M}_{3,k} &=&\bm{l}_{k}(\bm{P}_{2}\bm{P}_{1}) \oplus \bm{r}_{k}(\bm{P}_{1}) \oplus \bm{l}_k \oplus
                \bm{L}_{0,k}, \\
\bm{N}_{3,k} &=& \bm{l}_{k}(\bm{P}_{2}) \oplus\bm{r}_k \oplus \bm{R}_{0,k} \oplus  \bm{L}_{0,k}(\bm{P}_{0}).
\IEEEstrut
\end{IEEEeqnarraybox}
\right.
\end{equation*}
The above equations mean that $\bm{M}_{3,k}$, $\bm{N}_{3,k}$ and $\{\bm{P}_{l}\}_{l=0}^2$ can work together to recover the $k$-th least significant bit plane of any other cipher-image encrypted with the same secret key, $[\bm{L}^{\star}_{0,k}, \bm{R}^{\star}_{0,k}]$, for $k=1\sim n$.

To verify the above analysis, some similar experiments are made with $K_0=1234567/(2^{32}-1)$, $T=3$ and $A=64$. First, a chosen plain-image is composed
by combining the left half part of Fig.~\ref{figure:differentialattackR1}a) and the right half part of Fig.~\ref{figure:differentialattackR1}b), which makes
the special differential files satisfying Eq.~(\ref{eq:cancelpermute}) can be generated. Then, the three plain-images shown in Figs.~\ref{figure:differentialattackR1}a), b), Fig.~\ref{figure:differentialattackR3}a) and a plain-image ``Airplane" are encrypted with the same secret key, and the results are shown in Figs.~\ref{figure:differentialattackR3}b), c), d), e), respectively. With the three pairs of plain-images and cipher-images, some information about the secret key is obtained to decrypt the cipher-image shown in Fig.~\ref{figure:differentialattackR3}e) and the result is shown in Fig.~\ref{figure:differentialattackR3}f). It is counted that the $6$ least significant bit planes of the image shown in Fig.~\ref{figure:differentialattackR3}f) are identical with the counterparts of the corresponding plain-image also.

\begin{figure}[!htb]
\centering
\begin{minipage}{\figwidth}
\includegraphics[width=\textwidth]{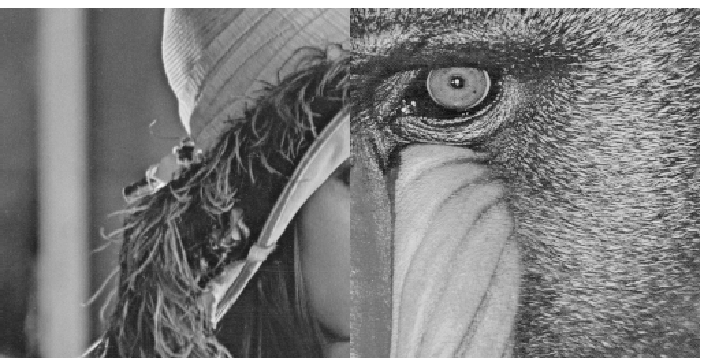}
a)
\end{minipage}
\begin{minipage}{\figwidth}
\includegraphics[width=\textwidth]{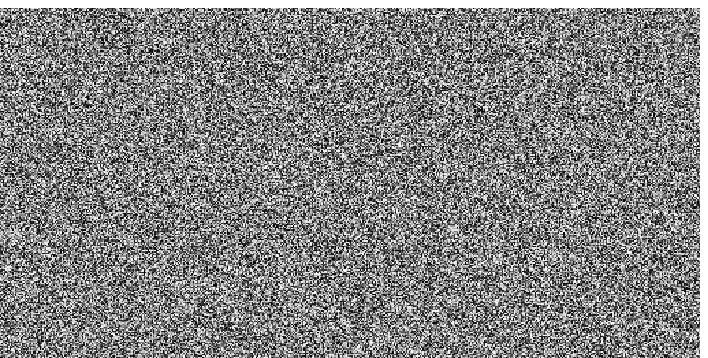}
b)
\end{minipage}\\
\begin{minipage}{\figwidth}
\includegraphics[width=\textwidth]{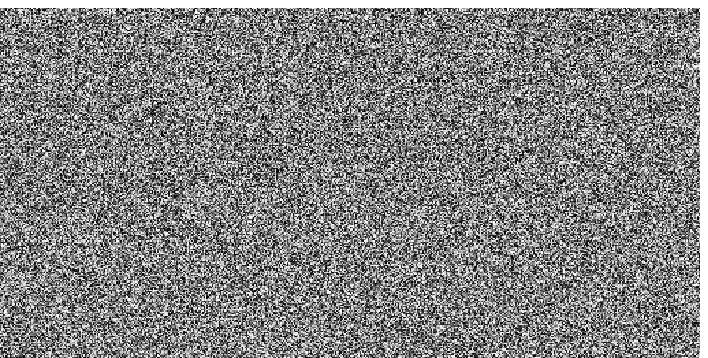}
c)
\end{minipage}
\begin{minipage}{\figwidth}
\includegraphics[width=\textwidth]{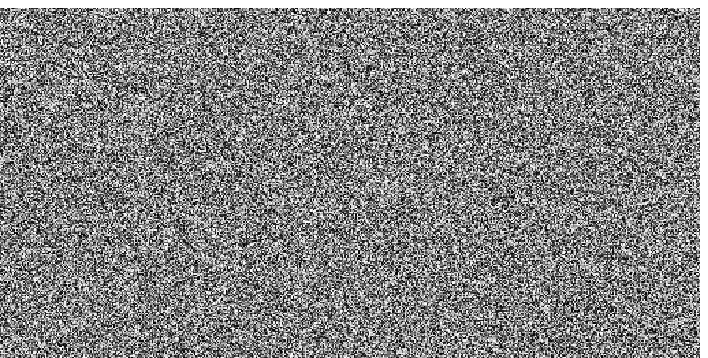}
d)
\end{minipage}\\
\begin{minipage}{\figwidth}
\includegraphics[width=\textwidth]{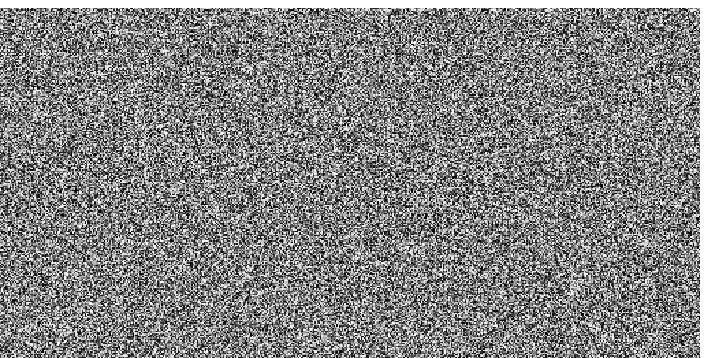}
e)
\end{minipage}
\begin{minipage}{\figwidth}
\includegraphics[width=\textwidth]{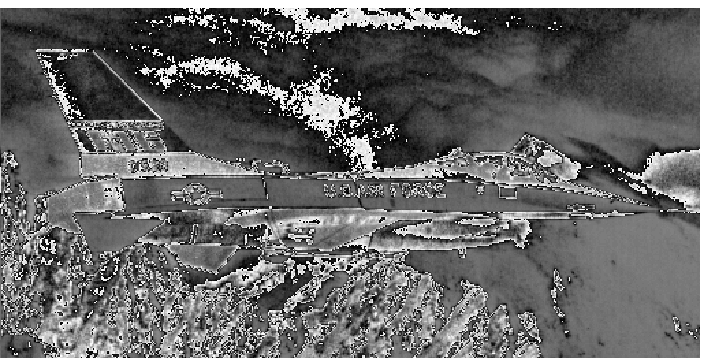}
f)
\end{minipage}
\caption{Differential attack on IEAS when $T=3$:
a) the constructed plain-image;
b) cipher-image of Fig.~\ref{figure:differentialattackR1}a);
c) cipher-image of Fig.~\ref{figure:differentialattackR1}b);
d) cipher-image of Fig.~\ref{figure:differentialattackR3}a);
e) cipher-image of the plain-image ``Airplane";
f) the recovered plain-image of Fig.~\ref{figure:differentialattackR3}e).}
\label{figure:differentialattackR3}
\end{figure}

\subsection{Breaking IEAS of higher rounds $(T\geq 4)$}

It is not hard to notice that there are some general approaches to breaking IEAS of different rounds. Here, we take breaking
the version of IEAS under four rounds as an example to illustrate how to implement differential attack on IEAS in a general way.

\begin{itemize}
\item \textit{Step 1) breaking position permutation}:

According to Property~\ref{Pro:differentialstructure}, one has
\begin{IEEEeqnarray}{rCl}
\bm{L}_{4,k}'   &=& \bm{R}'_{4, k}(\bm{P}_{3}) \oplus \bm{R}'_{3,k} \nonumber\\
                &=& \bm{L}'_{3, k}(\bm{P}_{3}) \oplus \bm{L}'_{2, k} \nonumber\\
                &=& \bm{R}'_{2, k}(\bm{P}_{3}) \oplus \bm{L}'_{2, k}(\bm{P}_2 \bm{P}_{3})  \oplus \bm{R}'_{2, k}(\bm{P}_{1}) \oplus  \bm{R}'_{1, k} \nonumber\\
                &=& \bm{R}'_{0,k}(\bm{P}_{3}) \oplus \bm{L}'_{0, k}(\bm{P}_0 \bm{P}_{3}) \oplus \bm{L}'_{0, k}(\bm{P}_2 \bm{P}_{3}) \oplus \bm{L}'_{1, k}(\bm{P}_1 \bm{P}_{2}\bm{P}_{3})  \nonumber \\
                && \: { }\oplus\bm{R}'_{0,k}(\bm{P}_{1}) \oplus \bm{L}'_{0, k}(\bm{P}_0 \bm{P}_{1})\oplus \bm{L}'_{0, k},
                \label{eq:differentialL4}
\end{IEEEeqnarray}
and $\bm{L}'_{1,k} = \bm{R}'_{0,k} \oplus \bm{R}'_{1,k}(\bm{P}_0)$, where $\bm{R}_{4,k}'$ is the $k$-th least significant bit plane of $\bm{R}_{4}'$.
Then, the problem become how to recover the permutation matrixes generated by Eq.~(\ref{eq:permutematrix}) by constructing some special differential plain-images.
\begin{itemize}
\item \textit{Determining $\bm{P}_1$ and $\bm{P}_3$ by choosing special $\bm{R}'_{0,k}$}

If $\bm{L}'_{0,k}$ is chosen of fixed value zero, one can get $\bm{L}'_{1,k} = \bm{R}'_{0,k}$. Substitute it into Eq.~(\ref{eq:differentialL4}), one has
\begin{equation}
\bm{L}_{4,k}'= \bm{R}'_{0, k}(\bm{P}_{1}) \oplus \bm{R}'_{0, k}(\bm{P}_{3}) \oplus \bm{R}'_{0, k}(\bm{P}_{1} \bm{P}_2 \bm{P}_3).
\label{eq:leftzeroR4}
\end{equation}
Assume a special differential image satisfy $\bm{L}'_{0,k}(i,j)\equiv 0$ and $\bm{R}'_{0,k}(i,j) = 0$ except that
\begin{equation}
\left\{\,
\begin{IEEEeqnarraybox}[][c]{rCl}
\IEEEstrut
\bm{R}'_{0,k}(i_1,j_1) &=& \alpha_1,\\
\bm{R}'_{0,k}(i_2,j_2) &=& \beta_1,
\IEEEstrut
\end{IEEEeqnarraybox}
\right.
\label{eq:condition1}
\end{equation}
where $\gcd (i_1j_2-i_2j_1,N)=1$ and $\alpha_1\neq \beta_1$. Observe Eq.~(\ref{eq:leftzeroR4}), one can see that one pixel of $\bm{R}'_{0, k}$ can influence at most three pixels of $\bm{L}'_{4,k}$. So, one can get $\binom{3}{1}\cdot\binom{3-1}{1} =6$ possible values of $(\bm{C}_1, \bm{C}_3, \bm{C}_1\bm{C}_2\bm{C}_3)$ by referring to Property~1. When condition of Proposition~1 exist, the matrix $(\bm{C}_1\bm{C}_2\bm{C}_3)$ can be recognized by checking which matrix whose elements are all greater than one\footnote{To simply analysis, the cases when $a, b\in\{0, 1\}$ and elements
of multiplication of three matrixes of set~(\ref{set:matrix}) are happen to be $(1\bmod N)$ are not discussed here.}. Since multiplication of two different matrixes of set~(\ref{set:matrix}) is not commutative when $(a+b)\neq 0$, $\bm{C}_1$ and $\bm{C}_3$ can be confirmed by checking whether $(\bm{C}_1^{-1}(\bm{C}_1\bm{C}_2\bm{C}_3)\bm{C}_3^{-1})$ has the form of the matrixes of set~(\ref{set:matrix}). Finally, the corresponding associated matrixes $\bm{P}_1$ and $\bm{P}_3$ can be obtained.

\begin{proposition}
When $a, b\not\in\{0, 1\}$, there is no 1's in the product of any three matrixes (including the same matrixes) of set~(\ref{set:matrix}).
\end{proposition}
\begin{proof}
When $a, b\not\in\{0, 1\}$, every element of the four matrixes in set~(\ref{set:matrix}) is greater than or equal to one. According to
multiplication rule of matrix, it can easily conclude that the proposition held.
\end{proof}

\item \textit{Determining $\bm{P}_0$ and $\bm{P}_2$ by choosing special $\bm{L}'_{0,k}$}

If $\bm{R}'_{0,k}$ is chosen of fixed value zero, it is easy to get
\begin{equation*}
\bm{R}_{4,k}'=\bm{L}'_{0, k}(\bm{P}_{0}) \oplus \bm{L}'_{0, k}(\bm{P}_{2}) \oplus \bm{L}'_{0, k}(\bm{P}_0 \bm{P}_1 \bm{P}_{2}).
\end{equation*}
Construct another special differential image satisfying $\bm{R}'_{0,k}(i,j) \equiv 0$ and $\bm{L}'_{0,k}(i,j)=0$ except that
\begin{equation}
\left\{\,
\begin{IEEEeqnarraybox}[][c]{rCl}
\IEEEstrut
\bm{L}'_{0,k}(i_1,j_1) &=& \alpha_2,\\
\bm{L}'_{0,k}(i_2,j_2) &=& \beta_2,
\IEEEstrut
\end{IEEEeqnarraybox}
\right.
\label{eq:condition2}
\end{equation}
where $\gcd (i_1j_2-i_2j_1,N) =1$ and $\alpha_2\neq \beta_2$. Then one can use the same method mentioned above to get the permutation matrixes $\bm{P}_0$ and $\bm{P}_2$.
\end{itemize}

\item \textit{Step 2) breaking value substitution}:

From Eq.~(\ref{eq:maskright}) and Property~\ref{Pro:differentialstructure}, one can get
\begin{equation}
\bm{l}_k  = \bm{V}_{5,k} \oplus \bm{R}_{4,k}   \label{eq:expressLKR4}
\end{equation}
and
\begin{IEEEeqnarray}{rCl}
\bm{r}_k  &=& \bm{V}_{4,k} \oplus \bm{L}_{4,k} \nonumber\\
            &=& \bm{V}_{4,k} \oplus \bm{R}_{3,k} \oplus \bm{R}_{4,k}(\bm{P}_3)\nonumber\\
            &=& \bm{V}_{4,k} \oplus \bm{V}_{2,k} \oplus \bm{L}_{2,k}\oplus \bm{R}_{4,k}(\bm{P}_3)\nonumber\\
            &=& \bm{V}_{4,k} \oplus \bm{V}_{2,k} \oplus \bm{R}_{1,k} \oplus \bm{R}_{2,k}(\bm{P}_1)\oplus \bm{R}_{4,k}(\bm{P}_3)\nonumber\\
            &=& \bm{V}_{4,k} \oplus \bm{V}_{2,k} \oplus \bm{V}_{0,k} \oplus \bm{L}_{0,k} \oplus   \bm{V}_{1,k}(\bm{P}_1) \oplus \bm{L}_{1,k}(\bm{P}_1) \oplus \bm{R}_{4,k}(\bm{P}_3)\nonumber\\
            &=& \bm{V}_{4,k} \oplus \bm{V}_{2,k} \oplus \bm{V}_{0,k} \oplus \bm{L}_{0,k} \oplus   \bm{V}_{1,k}(\bm{P}_1) \oplus \bm{R}_{0,k}(\bm{P}_1) \oplus \bm{R}_{1,k}(\bm{P}_0\bm{P}_1)\oplus \bm{R}_{4,k}(\bm{P}_3)\nonumber\\
            &=& \bm{V}_{4,k} \oplus \bm{V}_{2,k} \oplus \bm{V}_{0,k} \oplus \bm{L}_{0,k} \oplus   \bm{V}_{1,k}(\bm{P}_1) \oplus \bm{R}_{0,k}(\bm{P}_1) \nonumber \\
            && \oplus\: \bm{V}_{0,k}(\bm{P}_0\bm{P}_1)\oplus \bm{L}_{0,k}(\bm{P}_0\bm{P}_1) \oplus \bm{R}_{4,k}(\bm{P}_3).\nonumber
\end{IEEEeqnarray}
Substitute $\bm{R}_{4,k}$ obtained in Eq.~(\ref{eq:expressLKR4}) into the above equation, one has
\begin{multline}
  \bm{V}_{5,k}(\bm{P}_3) \oplus \bm{V}_{4,k} \oplus \bm{V}_{2,k} \oplus \bm{V}_{0,k}  \oplus  \bm{V}_{1,k}(\bm{P}_1) \oplus  \bm{V}_{0,k}(\bm{P}_0\bm{P}_1) \\
    =  \bm{l}_k(\bm{P}_3) \oplus \bm{r}_k  \oplus \bm{L}_{0,k} \oplus \bm{R}_{0,k}(\bm{P}_1) \oplus \bm{L}_{0,k}(\bm{P}_0\bm{P}_1).
 \label{eq:resultR41}
\end{multline}

Referring to Eq.~(\ref{eq:expressR2KR3}) and  Eq.~(\ref{eq:resulttoR22}), one can get
\begin{IEEEeqnarray}{rCl}
\bm{R}_{4,k}  &=& \bm{V}_{3,k} \oplus \bm{L}_{3,k} \nonumber\\
            &=& \bm{V}_{3,k} \oplus \bm{R}_{2,k} \oplus \bm{R}_{3,k}(\bm{P}_2)\nonumber\\
           &=& \bm{V}_{3,k} \oplus \bm{V}_{1,k} \oplus \bm{R}_{0,k}\oplus\bm{V}_{0,k}(\bm{P}_0)
           \oplus\bm{L}_{0,k}(\bm{P}_0) \oplus\bm{R}_{3,k}(\bm{P}_2)\nonumber\\
            &=& \bm{V}_{3,k} \oplus \bm{V}_{1,k} \oplus \bm{R}_{0,k}\oplus\bm{V}_{0,k}(\bm{P}_0)
           \oplus\bm{L}_{0,k}(\bm{P}_0) \oplus \bm{V}_{2,k}(\bm{P}_{2}) \oplus \bm{V}_{0,k}(\bm{P}_{2})\nonumber\\
             && \oplus\:\bm{L}_{0,k}(\bm{P}_{2}) \oplus \bm{V}_{1,k}(\bm{P}_1\bm{P}_2) \oplus \bm{R}_{0}(\bm{P}_{1}\bm{P}_{2})\oplus \bm{V}_{0}(\bm{P}_0\bm{P}_{1}\bm{P}_{2})\oplus \bm{L}_{0}(\bm{P}_0\bm{P}_{1}\bm{P}_{2}).\nonumber
\end{IEEEeqnarray}
Hence Eq.~(\ref{eq:expressLKR4}) become
\begin{multline}
  \bm{V}_{5,k} \oplus \bm{V}_{3,k} \oplus \bm{V}_{1,k}\oplus\bm{V}_{0,k}(\bm{P}_0)\oplus \bm{V}_{2,k}(\bm{P}_{2})\oplus \bm{V}_{0,k}(\bm{P}_{2})\oplus \bm{V}_{1,k}(\bm{P}_1\bm{P}_2)\oplus \bm{V}_{0}(\bm{P}_0\bm{P}_{1}\bm{P}_{2})\\
  = \bm{l}_k \oplus \bm{R}_{0,k} \oplus\bm{L}_{0,k}(\bm{P}_0)\oplus \bm{L}_{0,k}(\bm{P}_{2}) \oplus \bm{R}_{0}(\bm{P}_{1}\bm{P}_{2})\oplus \bm{L}_{0}(\bm{P}_0\bm{P}_{1}\bm{P}_{2}).
 \label{eq:resultR42}
\end{multline}
Substitute $\bm{L}_{0,k}(\bm{P}_0\bm{P}_1)$ obtained in Eq.~(\ref{eq:resultR41}) into Eq.~(\ref{eq:resultR42}) yields
\begin{multline}
  \bm{V}_{5,k}(\bm{P}_3\bm{P}_2) \oplus \bm{V}_{4,k}(\bm{P}_2)   \oplus
   \bm{V}_{5,k} \oplus \bm{V}_{3,k} \oplus \bm{V}_{1,k}\oplus\bm{V}_{0,k}(\bm{P}_0) \\
   =\bm{l}_k(\bm{P}_3\bm{P}_2) \oplus \bm{r}_k(\bm{P}_2) \oplus \bm{l}_k \oplus \bm{R}_{0,k} \oplus \bm{L}_{0,k}(\bm{P}_0).\label{eq:resultR43}
\end{multline}
Substitute $\bm{L}_{0,k}(\bm{P}_0)$ obtained in Eq.~(\ref{eq:resultR43}) into Eq.~(\ref{eq:resultR41}), one can get
\begin{multline}
  \bm{V}_{5,k}(\bm{P}_3\bm{P}_2\bm{P}_1) \oplus \bm{V}_{4,k}(\bm{P}_2\bm{P}_1)   \oplus \bm{V}_{5,k}(\bm{P}_1)
  \oplus \bm{V}_{3,k}(\bm{P}_1)\oplus \bm{V}_{5,k}(\bm{P}_3) \oplus \bm{V}_{4,k} \oplus \bm{V}_{2,k} \oplus \bm{V}_{0,k}
   \\
   =  \bm{l}_k(\bm{P}_3\bm{P}_2\bm{P}_1) \oplus \bm{r}_k(\bm{P}_2\bm{P}_1) \oplus \bm{l}_k(\bm{P}_1)\oplus \bm{l}_k(\bm{P}_3) \oplus \bm{r}_k  \oplus \bm{L}_{0,k}.
\label{eq:resultR44}
\end{multline}

\item \textit{Step 3) decrypting another cipher-image encrypted with the same secret key}:

Since both Eq.~(\ref{eq:resultR43}) and Eq.~(\ref{eq:resultR44}) exist for any pair of plain-image and its corresponding cipher-image, so one can get
\begin{equation*}
\left\{\,
\begin{IEEEeqnarraybox}[][c]{rCl}
\IEEEstrut
\bm{L}^{\star}_{0,k} &=& \bm{l}^{\star}_{k}(\bm{P}_{3}\bm{P}_{2}\bm{P}_{1}) \oplus \bm{l}^{\star}_{k}(\bm{P}_{1}) \oplus  \bm{l}^{\star}_{k}(\bm{P}_{3}) \oplus \bm{r}^{\star}_{k}(\bm{P}_{2}\bm{P}_{1}) \oplus \bm{r}^{\star}_k
               \oplus  \bm{M}_{4,k}, \\
\bm{R}^{\star}_{0,k} &=& \bm{l}^{\star}_{k}(\bm{P}_{3}\bm{P}_{2}) \oplus \bm{l}^{\star}_k \oplus \bm{r}^{\star}_{k}(\bm{P}_{2}) \oplus  \bm{L}^{\star}_{0,k}(\bm{P}_{0}) \oplus
\bm{N}_{4,k}
\IEEEstrut
\end{IEEEeqnarraybox}
\right.
\end{equation*}
where
\begin{equation*}
\left\{
\begin{IEEEeqnarraybox}[][c]{rCl}
\IEEEstrut
\bm{M}_{4,k} &=& \bm{l}_{k}(\bm{P}_{3}\bm{P}_{2}\bm{P}_{1})\oplus \bm{r}_{k}(\bm{P}_{2}\bm{P}_{1}) \oplus \bm{l}_{k}(\bm{P}_{1}) \oplus  \bm{l}_{k}(\bm{P}_{3})  \oplus \bm{r}_k
               \oplus  \bm{L}_{0,k}, \\
\bm{N}_{4,k} &=& \bm{l}_{k}(\bm{P}_{3}\bm{P}_{2}) \oplus \bm{r}_{k}(\bm{P}_{2}) \oplus \bm{l}_k \oplus
\bm{R}_{0,k} \oplus  \bm{L}_{0,k}(\bm{P}_{0}).
\IEEEstrut
\end{IEEEeqnarraybox}
\right.
\end{equation*}
The above equation means that $\{\bm{M}_{4,k}\}_{k=1}^n$, $\{\bm{N}_{4,k}\}_{k=1}^n$, and $\{\bm{P}_{l}\}_{l=0}^3$ can work together to recover the $n$ least significant bit planes of the right part of $\bm{l}^{\star}$ and $\bm{r}^{\star}$, $\{\bm{L}^{\star}_{0,k}\}_{k=1}^n$ and $\{\bm{R}^{\star}_{0,k}\}_{k=1}^n$.
\end{itemize}

To verify the above analysis, experiments are made with $K_0=1234567/(2^{32}-1)$, $T=4$, and $A=64$ or 128. First, two special known-images are generated by modifying the image shown in Fig.~\ref{figure:differentialattackR1}a) to make the differential images satisfy condition~(\ref{eq:condition2}). Due to similarity of the two constructed plain-image, only one of them is shown in Fig.~\ref{figure:differentialattackR4}a). Similarly, the other two special known-image are constructed by modifying the image shown in Fig.~\ref{figure:differentialattackR1}a).
The encryption result of the plain-image ``Airplane" is shown in Fig.~\ref{figure:differentialattackR4}b). With the five chosen plain-images, some information about the secret key is obtained to decrypt the cipher-image shown in Fig.~\ref{figure:differentialattackR4}b) and the result is shown in Fig.~\ref{figure:differentialattackR4}c).
When only $A$ is changed as 128, the recover image of the corresponding cipher-image of the plain-image ``Airplane" is shown in Fig.~\ref{figure:differentialattackR4}d)
Once again, the experiment results demonstrate that the breaking performance is mainly by the integer $n$ in Property 2.

\begin{figure}[!htb]
\centering
\begin{minipage}{\figwidth}
\includegraphics[width=\textwidth]{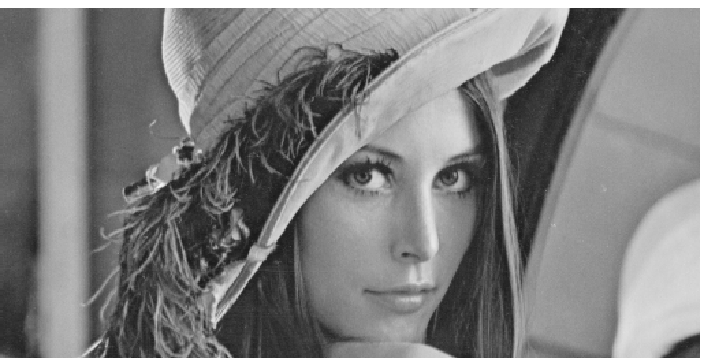}
a)
\end{minipage}
\begin{minipage}{\figwidth}
\includegraphics[width=\textwidth]{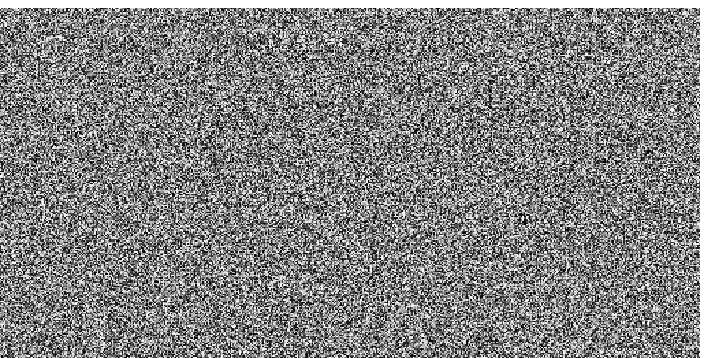}
b)
\end{minipage}\\
\begin{minipage}{\figwidth}
\includegraphics[width=\textwidth]{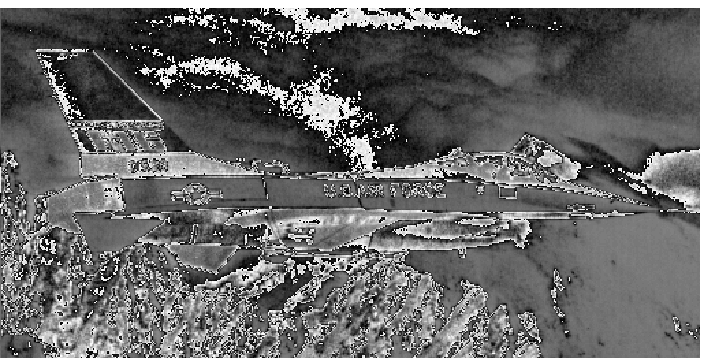}
c)
\end{minipage}
\begin{minipage}{\figwidth}
\includegraphics[width=\textwidth]{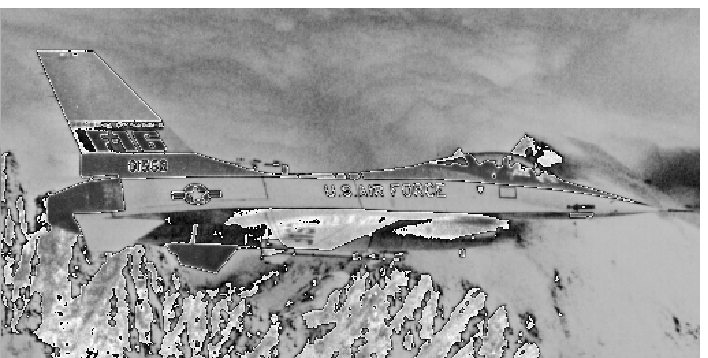}
d)
\end{minipage}
\caption{Differential attack on IEAS when $T=4$:
a) chosen plain-image;
b) cipher-image of plain-image ``Airplane";
c) the recovered plain-image with $A=64$;
d) the recovered plain-image with $A=128$.}
\label{figure:differentialattackR4}
\end{figure}

\section{Some other security defects of IEAS}

To make cryptanalysis on IEAS more complete, some other security defects of IEAS are given in this section.
\begin{itemize}

\item{The key space of IEAS is not big enough}

In \cite[Sec.~4]{Zhang:ImageCrypt:SCSF07}, it is claimed that key space of IEAS is $2^{32(T+2)}$ since PRNS $\{K_{l}\}_{l=0}^{T+1}$ has
$32(T+2)$ bits. However, this it not true since $\{K_{l}\}_{l=0}^{T+1}$ is generated by the Logistic map under initial condition $K_0$, who
has only $n_0$ unknown bits, where $n_0$ is precision length of computer. In fact, permutation matrixes $\{\bm{P}_l\}_{l=0}^{T-1}$ and mask matrixes $\{\bm{V}_l\}_{l=0}^{T+1}$
can compose an equivalent secret key of IEAS, and $\{\bm{P}_l\}_{l=0}^{T-1}$ has only $4^T$ possible cases. Since generation of $\{\bm{P}_l\}_{l=0}^{T-1}$
is also controlled by $\{K_{l}\}_{l=0}^{T+1}$, we can conclude that the real key space of IEAS is only $2^{n_0}\cdot T=2^{n_0}T$. In \cite{Zhang:ImageCrypt:SCSF07},
$n_0=32$, so the key space of IEAS is less than $2^{32}16=2^{36}$ considering $T\le 16$. Even computation precision of 64 bits is used, the key space is
only $2^{68}$, which is lower than expected size of a secure cipher, $2^{128}$, much.

\item{Insufficient sensitivity with respect to change of plain-image}

As well known in cryptography, sensitivity of ciphertext with respect to changes of plaintext is a very important property
measuring a secure encryption scheme. This property is especially important for secure image encryption schemes since a plain-image
and its watermarked version are often encrypted in the same time. In \cite[Sec.~4.2]{Zhang:ImageCrypt:SCSF07}, it was claimed that
IEAS satisfy the property well. However, IEAS fail to do it much due to the following points.
\begin{itemize}
\item The sole nonlinear operation is only used to expand PRBS, and no nonlinear operation, like S-box, is involved of handling plain-image;

\item There is no any operation generating carry bit toward lower level in the whole scheme, so a bit of plain-image can only influence the
bits at higher bit planes in the cipher-image;

\item If $2^n$ $(1\le n\le 7)$ divides variable $A$ in Eq.~(\ref{eq:difussion}), any change of the bits in the $k$-th bit plane of plain-image will only affect
the bits in the same bit plane of cipher-image for $k=1\sim n$.
\end{itemize}

\item{Superior performance of IEAS is questionable}

The cryptanalysis presented in the above section is based on the precondition of Property~2, namely $2^n$ $(1\le n\le 7)$ divides variable $A$ in Eq.~(\ref{eq:difussion}). This means that IEAS would become robust against the proposed attack if $A$ is odd. Under this condition, Property~2 is still
exist with some probability. So, the proposed attack maybe still valid with a little higher complexity. To show inferior performance of IEAS
is undoubted in any cases, IEAS is compared with its analogue, DES. The encryption complexity of DES on 128 plain-bits and the widely recognized
robustness of DES against differential attack under some rounds are shown in Table~\ref{table:comparision} \cite{Biham:DES:Cryptology97,Biham:EnhancingDL:LecNotes02}. In contrast, encryption complexity of IEAS on the same data and robustness against differential
attack are shown in Table~\ref{table:comparision} also. Although the details deriving attack complexity of IEAS of round number is larger than four are
not given here, one can conclude confidently that IEAS is much weaker than DES now.

\begin{table}[!htbp]
\centering
\caption{Comparison between IEAS and DES in terms of complexity of encrypting $128$ plain-bits and robustness against
differential attack, where CP and KP denote chosen plaintexts and known plaintexts, respectively.}
\begin{tabular}{c|c|c|c|c|c|c}
\hline
\multirow{3}{35pt}{Round Number} & \multicolumn{2}{c|}{Complexity}& \multicolumn{4}{c}{Attack} \\\cline{2-7}
           & \multirow{2}{*}{DES}          & \multirow{2}{*}{IEAS} & \multicolumn{2}{c|}{Data}    & \multicolumn{2}{c}{Success Rate}\\\cline{4-7}
           &                  &            &{DES}           &{IEAS}                     & {DES}   & {IEAS}\\ \hline
\hline 1   &$O(2^{9})$        &$O(2^{12})$ & $O(1)$CP       &2CP                        &$100\%$  &$100\%$\\
\hline 2   &$O(2^{10})$       &$O(2^{13})$ & $O(1)$CP       &2CP                        &$100\%$  &$100\%$\\
\hline 3   &$O(2^{10})$       &$O(2^{13})$ & $O(1)$CP       &3CP                        &$100\%$  &$100\%$\\
\hline 4   &$O(2^{11})$       &$O(2^{14})$ & $2^{4}$CP      &5CP                        &$100\%$  &$100\%$\\
\hline 12  &$O(2^{12})$       &$O(2^{15})$ & $2^{44}$KP     &$14$KP                     &$10\%$   &$100\%$\\
\hline 13  &$O(2^{12})$       &$O(2^{15})$ & $2^{45}$KP     &$14$KP                     &$10\%$   &$100\%$\\
\hline 16  &$O(2^{13})$       &$O(2^{16})$ & $2^{50}$KP     &$14$KP                     &$51.3\%$ &$100\%$\\
\hline
\end{tabular}
\label{table:comparision}
\end{table}
\end{itemize}

\section{Conclusion}

The security of an image encryption algorithm called IEAS, a block cipher composing of multiple rounds, was studied comprehensively in this paper.
Some properties of IEAS are derived to support differential attack on it when its key parameter is even. The detailed approaches for breaking IEAS,
when round number is less than five, are presented and can be easily extended to break the version of IEAS of higher rounds. In addition, it is found that
encryption results of IEAS is not sensitive with respect to changes of plain-image and its key space is not big enough. Cryptanalysis of IEAS shown in this paper
and comparison between IEAS and DES demonstrate IEAS is not attractive secure image encryption scheme and should not be used in applications requiring high level of
security.

\section*{Acknowledgement}

This research was supported by the National Natural Science Foundation of China (No.~61100216), Scientific Research Fund of Hunan Provincial Education Department (Nos.~11B124,~2011FJ2011), and start-up funding of Xiangtan University (Nos.~10QDZ39,~10QDZ40).

\bibliographystyle{elsarticle-num}
\bibliography{SCSF}
\end{document}